\documentclass[acmtocl,acmnow]{acmtrans2m}

\usepackage{latexsym}
\usepackage{amssymb}
\usepackage{amsmath}
\usepackage{stmaryrd}
\usepackage{enumerate}

\newtheorem{theorem}{Theorem}[section]

\newtheorem{corollary}[theorem]{Corollary}
\newtheorem{proposition}[theorem]{Proposition}

\newdef{definition}[theorem]{Definition}
\newdef{remark}[theorem]{Remark}
\newdef{example}{Example}

\def\bbN{\mathbb{N}}

\def\half{{\frac{1}{2}}}

\def\stra{\mathfrak{A}}

\def\floor #1{\mbox{$\lfloor #1 \rfloor$}}

\def\vq{{\bar{q}}}

\def\vx{{\bar{x}}}
\def\vy{{\bar{y}}}

\def\nram#1{\textsc{NRAM}[#1]}

\def\ntime#1{\textsc{NTIME}[#1]}
\def\cone{\textsc{co-NE}}
\def\ne{\textsc{NE}}
\def\conp{\textsc{co-NP}}
\def\np{\textsc{NP}}
\def\p{\textsc{P}}

\def\soe{\textsc{SO}\exists}

\def\soevar#1{\textsc{SO}\exists\mathrm{(var}\ #1\mathrm{)}}

\def\fspec{\textsc{F-Spec}}
\def\spec{\textsc{Spec}}
\def\cospec{\textsc{Co-Spec}}

\def\bin{\textsf{BINARY}}

\def\arity{\textrm{arity}}

\def\relADD{\textsf{ADD}}
\def\relMUL{\textsf{MUL}}
\def\relPROJECT{\textsf{PROJECT}}
\def\relDOUBLE{\textsf{DOUBLE}}
\def\relHALF{\textsf{HALF}}
\def\relDIV{\textsf{DIV}}
\def\relBIT{\textsf{BIT}}
\def\relISR{\textsf{IS-R}}
\def\relLESSR{\textsf{LESS-R}}

\def\relLESSRR{\textsf{LESS-R2}}
\def\relSYMBOL{\textsf{SYMBOL}}
\def\relSTATE{\textsf{STATE}}
\def\relMIN{\textsf{MIN}}
\def\relMAX{\textsf{MAX}}
\def\suc{\textsf{SUC}}
\def\relsucx{\textsf{SUCX}}
\def\relsucy{\textsf{SUCY}}
\def\relMINx{\textsf{MINX}}
\def\relMINy{\textsf{MINY}}
\def\relcyc{\textsf{RCYC}}

\def\relINPUT{\textsf{INPUT}}

\def\IFF{\Leftrightarrow}

\def\OMIT #1{}

\title{On the variable hierarchy of first-order spectra}

\author{Eryk Kopczy\'nski
\\
University of Warsaw
\\
\\
Tony Tan\\
Hasselt University and Transnational University of Limburg}
            
\begin{abstract}
The {\em spectrum} of a first-order logic sentence 
is the set of natural numbers that are cardinalities of its finite models.
In this paper we study the hierarchy of first-order spectra
based on the number of variables.
It has been conjectured that it collapses to three variable.
We show the opposite: it forms an infinite hierarchy.
However, despite the fact that more variables
can express more spectra, 
we show that to establish whether the class of first-order spectra is closed under complement,
it is sufficient to consider sentences using only three variables and binary relations.
\end{abstract}

\category{F.1.3}{Complexity Measures and Classes}{Complexity hierarchies}
\category{F.4.1}{Mathematical Logic}{Finite model theory}
         
\terms{Theory}         
            
\keywords{first-order spectra, bounded number of variables, non-deterministic exponential time}

\begin{document}

\maketitle

\section{Introduction}
\label{sec:intro}

The spectrum of a first-order sentence $\Phi$ (with the equality predicate), 
denoted by $\spec(\Phi)$,
is the set of natural numbers that are cardinalities of finite models of $\Phi$.
Or, more formally, $\spec(\Phi) = 
\{n \mid \Phi$ has a model with universe of cardinality $n\}$.
A set is a {\em spectrum},
if it is the spectrum of a first-order sentence.
We let $\spec$ to denote the class of all spectra.
Without the equality predicate, it is known that
if a sentence has a model of cardinality $n$,
then it also has a model of cardinality $n+1$.

The notion of the {\em spectrum} was introduced by Scholz, 
where he also asked whether there exists a 
necessary and sufficient condition for a set to be a spectrum~\cite{Scholz52}.
Since its publication, Scholz's question and many of its variants
have been investigated by many researchers for the past 60 years.
Arguably, one of the main open problems on spectra is 
the one asked by Asser,
known as {\em Asser's conjecture},
whether the complement of a spectrum is also a spectrum~\cite{Asser55}.

Though seemingly unrelated, it turns out that
the notion of spectra has a tight connection with complexity theory.
In fact, Asser's conjecture is shown to be equivalent to 
the problem $\ne$ vs. $\cone$\footnote{$\ne$ is the class of languages accepted
by a non-deterministic (and possibly multi-tape) Turing machine with run time
$O(2^{kn})$, for some constant $k > 0$.},
when Jones and Selman, as well as Fagin independently
showed that a set of integers is a spectrum if and only if
its binary representation is in $\ne$~\cite{JS74,Fagin73,Fagin}.
It also immediately implies that if
Asser's conjecture is false,
i.e., there is a spectrum whose complement is not a spectrum,
then $\np\neq\conp$, hence $\np\neq \p$.

In this paper we study the following hierarchy of spectra,
which we call the variable hierarchy:
For every integer $k \geq 1$, define
\begin{eqnarray*}
\spec_{k} & = & \{\spec(\Phi)\mid \Phi \ \mbox{uses only up to} \ k \ \mbox{variables}\}
\end{eqnarray*}
Obviously we have 
$\spec_{1}\subseteq \spec_2 \subseteq \cdots$.
It was conjectured that
the variable hierarchy collapses to three variables,
due to the fact that three variables are enough
to describe the computation of a Turing machine.
For more discussion on this conjecture, we refer the reader
to a recent survey by Durand, et. al~\cite{DJMM12}.

In this paper we show the opposite:
The variable hierarchy has an infinite number of levels.
That is,
for every $k \geq 3$, $\spec_{k} \subsetneq \spec_{2k+2}$
(Corollary~\ref{cor:var-hierarchy}).
Here we should note that it is already known that 
$\spec_{1}\subsetneq \spec_2 \subsetneq \spec_3$.
More discussion is provided in the next section.

Our proof follows from the following observations.
\begin{itemize}\itemsep=0pt
\item
To describe a computation of a non-deterministic Turing machine with runtime
$O(N^k)$ -- for a fixed integer $k \geq 1$ -- 
with a first-order sentence acting on a structure of cardinality $N$, 
$2k+1$ variables are sufficient.
\item
Conversely, for each first-order sentence $\Phi$ with $k$ variables, 
checking whether a structure of cardinality $N$ is a model of $\Phi$ 
can be done on a non-deterministic Turing machine 
in time $O(N^k(\log N)^2)$~\cite{Grandjean84,Grandjean85,GrandjeanO04}.
\end{itemize}
Curiously, despite the infinity of the variable hierarchy,
by standard padding argument, our proof implies that 
the class of first-order spectra is closed under complement if and only if
the complement of every spectrum of three-variable sentence (using only binary relations)
is also a spectrum (Corollary~\ref{cor:padding}).
This means that to settle Asser's conjecture,
it is sufficient to consider only three-variable sentences
using only binary relations.

This paper is organised as follows.
In Section~\ref{sec:related-works} we discuss some related results.
In Section~\ref{sec:easy-hierarchy}
we present a rather loose hierarchy:
for every integer $k \geq 3$, $\spec_{k}\subsetneq \spec_{2k+3}$.
Then in Section~\ref{sec:var-hierarchy}
we show that by more careful book-keeping,
we obtain a tighter hierarchy: 
For every integer $k \geq 3$, $\spec_{k}\subsetneq \spec_{2k+2}$.
In Section~\ref{sec:np-soe} we briefly discuss how our results
can be translated to the setting of generalised spectra.
We conclude with Section~\ref{sec:conclusion}.

\section{Related works}
\label{sec:related-works}

In this section we will briefly review the spectra problem
and discuss some related results.
We refer the reader to a recent survey by Durand, et. al. 
for a more comprehensive treatment on the spectra problem
and its history~\cite{DJMM12}.
Fagin's paper~\cite{Fagin93} covers nicely the relation between
the spectra problem
and finite model theory and its connection with descriptive complexity.

First, we remark that our result $\spec_{k}\subsetneq \spec_{2k+2}$,
for each integer $k\geq 3$ complements previous known result 
that $\spec_{1}\subsetneq \spec_{2} \subsetneq \spec_3$~\cite{DJMM12},
which can be proved as follows.
First, a model of first-order sentence with only one variable
remains a model after cloning elements, thus $\spec_1$ only
includes the empty set, and sets of form $\{n: n\geq k\}$.
In another paper we show that the class of spectra
of two-variable logic with counting quantifiers 
is exactly the class of semilinear sets,
and closed under complement~\cite{KT}. Using the same methods, one can
show that $\spec_2$ is the class of finite and cofinite sets,
thus separating $\spec_2$ from $\spec_1$.
On the other hand, three variables are enough to simulate
an arbitrary Turing machine, so it is not difficult to construct
a set in $\spec_3$ which is not even semilinear,
say, e.g., 
$\{n^2 \mid n \ \text{is the length of an accepting run of a Turing machine} \ M\}$,
hence, separating $\spec_3$ from $\spec_2$.

Related to the variable hierarchy is the arity hierarchy.
Let $\spec(\arity\ k)$ denote the spectra of first-order sentences
using only relations of arity at most $k$.
Fagin showed that 
if there exists $k$ such that $\spec(\arity\ k) =\spec(\arity \ k+1)$,
the arity hierarchy collapses to $k$~\cite{Fagin75}.

Lynch showed that $\ntime {N^k} \subseteq \spec(\arity\ k)$,
where $\ntime {N^k}$ denotes the class of sets of positive integers
(written in unary form) accepted by non-deterministic multi-tape
Turing machine in time $O(N^k)$, where $N$ is the input integer~\cite{Lynch82}.
The converse is still open and seems difficult.
A proof for $\spec(\arity\ k)\subseteq \ntime {N^k}$
seems to require that model checking for first-order sentences
(of arity $k$) on structures with universe of cardinality $N$
can be done in $\ntime {N^k}$.
However, a result by Chen, et. al.
states that checking whether a graph of $N$ vertices contains a $k$-clique,
which is of constant arity $2$,
cannot be done in time $O(N^{o(k)})$ unless the exponential time hypothesis fails~\cite{ChenHKX04,ChenHKX06,ImpagliazzoP99}.

Another body of related works is
those by Grandjean, Olive and Pudlak
which established the variable hierarchy
for spectra of sentences using relation and function 
symbols~\cite{Grandjean84,Grandjean85,Grandjean90,GrandjeanO04,Pudlak75}.
Let $\fspec_k$ denote the spectra of first-order sentences using up to $k$ variables
with vocabulary consisting of relation and function symbols,
and $\fspec (k\forall)$ denote the restriction of $\fspec_k$
to sentences written in prenex normal form with universal quantifiers only
and using only $k$ variables.
In his series of papers, Grandjean showed that
$\nram {N^k} = \fspec(k\forall)$, for each positive integer $k$,
where $\nram {N^k}$ denotes the class of sets of positive integers
accepted by non-deterministic RAM in time $O(N^k)$, 
and $N$ is the input integer~\cite{Grandjean84,Grandjean85,Grandjean90}.
By Skolemisation, it is shown that
$\fspec_k = \fspec(k\forall) = \nram {N^k}$, for all $k\geq 1$~\cite[Theorem~3.1]{GrandjeanO04}.
Combined with Cook's hierarchy of non-deterministic time~\cite{Cook73}
and the known inclusions 
$\ntime {T(n) \log T(n)} \subseteq \nram {T(n)} \subseteq {T(n) \log^2 T(n)}$,
for each function $T(n) \geq n$,
see~\cite{Grandjean85},
it implies $\fspec_{k} \subsetneq \fspec_{k+1}$, for all $k\geq 1$.

This does not imply our hierarchy here: $\spec_{k}\subsetneq \spec_{2k+2}$.
Obviously every function can be translated into a relation in first-order logic.
However, such translation requires at least one new variable for each function.
It is not clear whether there is a translation in which 
the number of new variables introduced
depends only on the arity of the functions, and not on the number of functions.
At this point we should also remark that $\fspec_k = \fspec(k\forall)$
can be much more expressive than $\spec_k$.
Take, for example, $k=1$.
The class $\spec_1$ consists of only empty set and sets of the form
$\{n,n+1,\ldots\}$,
whereas the class $\fspec(1\forall)$ contains 
PRIMES, the set of prime numbers~\cite{Grandjean90}.

\section{An easier hierarchy}
\label{sec:easy-hierarchy}

For a positive integer $N$,
we write $\bin(N)$ to denote its binary representation.
Correspondingly, for a set $A \subseteq \bbN$,
we write $\bin(A) \subseteq \{1,0\}^*$ to denote the set of 
the binary representations of the numbers in $A$.
To make comparison between languages and sets of positive integers,
for a function $T: \mathbb{N} \to \mathbb{N}$,
we define $\ntime {T(n)}$ to be the class of
sets of positive integers whose binary representations
are accepted by a non-deterministic (possibly multi-tape) Turing machine (NTM) 
with run time $O(T(n))$.
The class $\ne$ denotes $\bigcup_{k > 0} \ntime {2^{kn}}$.

Note that our definition implies that languages in $\ntime {T(n)}$
consist of strings that start with $1$.
This does not effect the generality of our results here.
For every language $L$, we can define $L' = \{1\} \cdot L$
and any Turing machine that accepts $L$ can be easily modified
to one that accepts $L'$ without any change in complexity.

In the following 
for a positive integer $n$,
we let $[n] = \{0,1,2,\ldots,n-1\}$. The proof of the following Proposition
will set a framework, which we will use again later in the proofs of
Theorems \ref{theo:ntime-spec} and \ref{theo:ntime-spec-half}.


\begin{proposition}
\label{prop:three-var}
$\ntime {2^{n}} \subseteq \spec_{3}$.
More precisely, for every set of positive integers $A$
where $\bin(A) \in \ntime {2^{n}}$,
there is a first-order sentence $\Phi$ 
using only three variables and binary relations such that
$\spec(\Phi) = A$.
\end{proposition}
\begin{proof}
The proof is via the standard encoding of an accepting run 
of an NTM with a square grid representing the space-time diagram.
Let $A$ be a set of positive integers, where $A\in \ntime {2^n}$. 
Let $M$ be a $t$-tape NTM accepting $\bin(A)$ in time $O(2^n)$ and space $O(2^n)$;
or, equivalently, for every $N \in A$, 
$M$ accepts $\bin(N)$ in time and space $O(N)$.
By linear speed-up~\cite[Theorem~2.2]{Papadimitriou94},
we can assume that $M$ accepts $\bin(N)$ in time and space $\leq N$.
This assumes that $N$ is big enough (greater than some $N_0$),
and this is not a problem for spectra -- numbers 
smaller than $N_0$ can always be considered on a case-by-case basis. 

For $N\in A$,
the accepting run of $M$ on $\bin(N)$ can be described
as a square-grid $[N]\times [N]$, where each point $(x,y) \in [N]\times [N]$
depicting cell $x$ in time $y$ is labelled according to the transitions of $M$.
We will construct a first-order sentence $\Phi$ such that 
the models of $\Phi$ are precisely such grids encoded as first-order structures 
of the universe $[N]$ with binary relations representing the labels of points 
$(x,y)\in [N]\times [N]$, and therefore, $\spec(\Phi)=A$.

The sentence $\Phi$ will be a conjunction
of \emph{axioms} which confirm that various parts of the model work as expected.
The proof will consist of two parts.
\begin{itemize}\itemsep=0pt
\item
Depicting the computation of $M$ with just three variables.

Essentially, in this part we want to describe that
the labels on the points $(x-1,y)$, $(x,y)$, $(x+1,y)$ and the labels on
its surrounding points $(x-1,y+1)$, $(x,y+1)$ and $(x+1,y+1)$
must ``match'' according to the transitions of $M$.
\item
Verifying that the input to $M$ is the binary representation of the
cardinality of the universe.
\end{itemize}
The details are as follows.

{\bf Depicting the computation of $M$ with just three variables.}
We first declare a successor $\suc$ and a total ordering $<$ on the
universe using three variables; this allows us to identify the universe
with $[N]$, and is done simply by adding the well-known 
total order and successor axioms to $\Phi$. 
The predicates
$\relMIN(x)$ and $\relMAX(x)$ state that
$x$ is the minimal and maximal element (0 and $N-1$), respectively. 

For a formula $\phi(x,y)$ with two free variables $x$ and $y$,
we take the third variable $z$, and define 
the operators $\Delta_h \phi(x,y)$, $\overline{\Delta}_h \phi(x,y)$ and $\Delta_v \phi(x,y)$,
where $h$ and $v$ stand for horizontal and vertical, respectively, as follows:
\begin{eqnarray*}
\Delta_h\phi (x,y) & := &  \forall z\ \suc(x,z) \Rightarrow \phi(z,y)
\\
\overline{\Delta}_h\phi (x,y) & := &  \forall z\ \suc(z,x) \Rightarrow \phi(z,y)
\\
\Delta_v\phi (x,y) & := &  \forall z\ \suc(y,z) \Rightarrow \phi(x,z)
\end{eqnarray*}
It is straightforward to see that for every $(x,y)$ when 
$x$ is not the minimal and the maximal elements
and $y$ is not the maximal elements,
\begin{itemize}\itemsep=0pt
\item
$\Delta_h \phi(x,y)$ holds if and only if $\phi(x+1,y)$ holds;
\item
$\overline{\Delta}_h \phi(x,y)$ holds if and only if $\phi(x-1,y)$ holds;
\item
$\Delta_v \phi(x,y)$ holds if and only if $\phi(x,y+1)$ holds.
\end{itemize}

Let the alphabet of $M$ be $\Sigma$, and $Q$ be the set of states  of $M$. 
We will require the following relations to simulate the machine:
\begin{itemize}\itemsep=0pt
\item 
$\relSYMBOL^i_a (x,y)$,
which holds if and only if the $x$-th cell of the $i$-th tape contains
the symbol $a$ at time $y$.
\item 
$\relSTATE^i_q (x,y)$, which holds if and only if 
the head on the $i$-th tape at time $y$
is over the $x$-th cell, and the state is $q$.
\end{itemize}

Now, to make sure that $\Phi$ depicts a computation of $M$ correctly,
we state the following:
{\em On every ``step'' $y=0,\ldots,N-1$,
if the heads are in states $q_1,\ldots,q_t$,
then for every cell $x=0,\ldots,N-1$,
the labels on $(x-1,y)$, $(x,y)$, $(x+1,y)$ and the labels on
$(x-1,y+1)$, $(x,y+1)$ and $(x+1,y+1)$ 
must ``match'' according to the transitions of $M$}.

Formally, it can be written as follows.
\begin{eqnarray*}
\bigwedge_{\vq=(q_1,\ldots,q_t) \in Q^t}
& \forall y &
\Bigg(
\Big(
\bigwedge_{1\leq i \leq t}
\exists x \ 
\relSTATE^i_{q_i} (x,y)
\Big)  \to 
\Big(
\bigwedge_{\phi} \
\forall x \
 \phi(x,y) 
\ \to \
\psi_{\phi,\vq}(x,y)
\Big)
\Bigg)
\end{eqnarray*}
where the intuitive meaning of $\phi$ and $\psi_{\phi,\vq}$ are as follows.
\begin{itemize}\itemsep=0pt
\item
The $\phi$ in the conjunction $\bigwedge_{\phi}$ runs through
all possible labels of $(x-1,y)$, $(x,y)$ and $(x+1,y)$,
where each $\phi$ is of form:
$$
\overline{\Delta}_h\ \ell ab_1(x,y)
\ \wedge \ \ell ab_2(x,y) \ \wedge \  \Delta_{h}\ \ell ab_3(x,y)
$$
Intuitively it means that $(x-1,y)$, $(x,y)$ and $(x+1,y)$
are labelled with $\ell ab_1$, $\ell ab_2$ and $\ell ab_3$, respectively,
where each $\ell ab_1$, $\ell ab_2$ and $\ell ab_3$ is 
a conjunction of the atomic relations $\relSTATE^i_q$,
and $\relSYMBOL^i_a$, as well as $\relMIN$ and $\relMAX$, 
and their negations to indicate whether $x$ or $y$ is the minimal or maximal element.
\item
The formula $\psi_{\phi,\vq}(x,y)$ is a disjunction of all possible labels
on the points $(x-1,y+1)$, $(x,y+1)$ and $(x+1,y+1)$
according to the transitions of $M$,
when the points $(x-1,y)$, $(x,y)$ and $(x+1,y)$ satisfy $\phi$ and
the states of the heads are $\vq = (q_1,\ldots,q_t)$.
Formally, $\psi_{\phi,\vq}(x,y)$ is of form:
\begin{eqnarray*}
\psi_{\phi,\vq}(x,y) & : = &
\Delta_{v} 
\Big(
\overline{\Delta}_h \psi_{\phi,\vq}'(x,y)
\ \wedge \
\psi_{\phi,\vq}''(x,y)
\ \wedge \
\Delta_h \psi_{\phi,\vq}'''(x,y)
\Big)
\end{eqnarray*}
where $\psi_{\phi,\vq}',\psi_{\phi,\vq}'',\psi_{\phi,\vq}'''$
are all the disjunctions of all possible labels on $(x-1,y+1)$, $(x,y+1)$ and $(x+1,y+1)$, respectively,
that are permitted by the transitions of $M$,
when the points $(x-1,y)$, $(x,y)$ and $(x+1,y)$ satisfy $\phi$ 
and the states of the heads are $\vq$.
\end{itemize}
Of course, we also have to state that 
{\em for every step $y = 0,\ldots,N-1$,
there are only $t$ heads}, i.e.
on every step $y = 0,\ldots,N-1$,
for every $i = 1,\ldots,t$,
there is exactly one cell $x$ where $(x,y)$ is labeled with $\relSTATE^i_{q}$. 
This is straightforward.

{\bf Verifying the input to the Turing machine.}
The input will be provided in binary. 
Recall that the elements of universe correspond to the numbers from $0$ to $N-1$.
We will need the following axioms.
\begin{itemize}\itemsep=0pt
\item 
The relation $\relDOUBLE(x,y)$ which holds if and only if $x=2y$.
It is defined inductively by $x=y=0$ and $(x-2) = 2(y-1)$.
$$
\forall x \forall y \
\left(
\begin{array}{c}
\relDOUBLE(x,y) 
\IFF
(\relMIN(x) \wedge \relMIN(y)) \ \vee \ 
\\
(\exists z \ (\suc(z,x) \wedge \exists x \ (\suc(x,z) \wedge \exists z (\suc(z,y) \wedge \relDOUBLE(x,z)))))
\end{array}
\right)
$$

\item 
The relation $\relHALF(x,y)$ which holds if and only if $x=\lfloor y/2 \rfloor$,
i.e. $y = 2x$ or $y = 2x + 1$.
$$
\forall x \forall y \
\left(
\relHALF(x,y) 
\IFF 
\relDOUBLE(y,x) \vee \exists z \ (\relDOUBLE(z,x) \wedge \suc(z,y) )
\right)
$$

\item 
The relation $\relDIV(x,y)$ which holds if and only if $x = \lfloor (N-1)/2^y\rfloor$.
It is defined inductively by $\floor {(N-1)/2^0} = N-1$ and 
$\floor{(N-1)/2^y} = \floor{\floor{(N-1)/2^{y-1}}/2}$).
$$
\forall x \forall y
\left(
\begin{array}{c}
\relDIV(x,y) \IFF (\relMAX(x) \wedge \relMIN (y))
\ \vee \
\\
\exists z(\suc(z, y)\wedge \exists y(\relDIV(y,z)\wedge \relHALF(x,y)))
\end{array}
\right)
$$

\item 
The relation $\relBIT(y)$ which holds 
if and only if the bit $b_y$ of the binary
representation $b_{N-1}\cdots b_1 b_0$ of $N-1$ is $1$, 
i.e., the integer $x = \floor{(N-1)/2^y}$ is odd.
$$
\forall y \
\left(
\relBIT (y)\IFF \exists x(\relDIV (x, y) \wedge \neg \exists z\relDOUBLE(x,z))
\right)
$$
\end{itemize}
Finally, notice that because the relation $\relBIT$ encodes 
the binary representation of $N-1$,
the relation denoted by $\relINPUT$ that encodes the input string, i.e., 
the binary representation of $N$, is defined by the following axiom:
$$
\exists x
\left(
\begin{array}{l}
\neg \relBIT (x) \ \wedge \  \relINPUT (x) \ \wedge
\left(
\begin{array}{l}
\forall y < x \ (\relBIT (y)\wedge\neg \relINPUT(y))\ \wedge
\\
\forall y > x \ (\relINPUT(y) \IFF \relBIT (y))
\end{array}
\right)
\end{array}
\right)
$$
This completes our proof of Proposition~\ref{prop:three-var}. 
\end{proof}

Proposition~\ref{prop:three-var} can be generalised to $\ntime {2^{kn}}$
as stated in the following theorem.
\begin{theorem}
\label{theo:ntime-spec}
For every integer $k \geq 1$, $\ntime {2^{kn}} \subseteq \spec_{2k+1}$.
\end{theorem}
\begin{proof}
The proof follows the same outline as the proof of Proposition \ref{prop:three-var}. 
Let $A$ be a set of positive integers such that $\bin(A) \in \ntime{2^{kn}}$ and 
$M$ be a $t$-tape NTM accepting $\bin(A)$ in time $N^k$ and space $N^k$.
So the space-time diagram is an $[N^k]\times [N^k]$ grid.

We identify numbers in $[N^k]$ with vectors 
$(p_k, p_{k-1} , \ldots, p_1) \in [N]^k$. 
The lexicographical successor relation 
$\suc(p_k,\ldots, p_1 , q_k ,\ldots, q_1)$ can be defined as
$1+ \sum_i p_i N^{i-1} = \sum_i q_i N^{i-1}$.


As in the proof of Proposition~\ref{prop:three-var},
the first-order sentence essentially states the following:
{\em On every ``step'' $\vy \in [N]^k$,
if the heads are in states $q_1,\ldots,q_t$,
then for every cell $\vx \in [N]^k$,
the labels on $(\vx'',\vy)$, $(\vx,\vy)$ and $(\vx',\vy)$ and the labels on
$(\vx'',\vy')$, $(\vx,\vy')$ and $(\vx',\vy')$ 
must ``match'' according to the transitions in $M$},
where $\vx'$ and $\vy'$ are the lexicographical successors of $\vx$ and $\vy$, respectively,
and $\vx''$ is the lexicographical predecessor of $\vx$.

Accordingly, the relations
$\relSYMBOL^i_a$ and $\relSTATE^i_q$ are of arity $2k$. 
The only significant difference is the shift operators
$\overline{\Delta}_h$, $\Delta_{h}$ and $\Delta_{v}$ 
which use only one extra variable, $z$, in their expansion.
Let $\vx= (x_k,\ldots,x_1)$ and $\vy=(y_k,\ldots,y_1)$.
The operator $\Delta_{h}$ is defined on any formula $\phi(\vx,\vy)$ as follows:
\begin{eqnarray*}
\Delta_h \phi(\vx, \vy) & := &
\bigvee^k_{i=2} \ \exists z \
\bigwedge^{i-1}_{j=1} 
\left(
\begin{array}{c}
\relMAX(x_j) \ \wedge \ \suc(x_i,z) \ \wedge \
\\ 
\exists x_1 (\relMIN(x_1) \wedge \phi(x_k,\ldots, x_{i+1}, z, x_1 ,\ldots, x_1 , \vy))
\end{array}
\right)
\\
& & \vee
\Big(\exists z \ \suc(x_1,z) \ \wedge \ \phi(x_k,\ldots,x_2,z,\vy)\Big)
\end{eqnarray*}
The operators $\overline{\Delta}_h$ and $\Delta_{v}$ can be defined in a similar manner.
As previously, it is straightforward to see that 
\begin{itemize}\itemsep=0pt
\item
$\Delta_{h} \phi(\vx,\vy)$ holds 
if and only if
$\phi(\vx',\vy)$ holds, where $\vx'$ is the lexicographical successor of $\vx$, and
\item
$\overline{\Delta}_{h} \phi(\vx,\vy)$ holds 
if and only if
$\phi(\vx',\vy)$ holds, where $\vx$ is the lexicographical successor of $\vx'$, and
\item 
$\Delta_{v} \phi(\vx,\vy)$ holds 
if and only if
$\phi(\vx,\vy')$ holds, where $\vy'$ is the lexicographical successor of $\vy$.
\end{itemize}
This completes the definition of the space-time grid structure, and thus completes
our proof of Theorem~\ref{theo:ntime-spec}.
\end{proof}

Next, we recall a result by Grandjean
which states that $k$-variable spectra, 
even if we use function symbols, can be computed effectively.

\begin{theorem}[\cite{Grandjean84,GrandjeanO04}]
\label{theo:spec-ntime}
For every 
\\
integer $k \geq 1$,
$\fspec_k \subseteq \ntime {n^2 2^{kn}}$.
\end{theorem}

Combining Theorems~\ref{theo:ntime-spec} and~\ref{theo:spec-ntime},
we obtain the following hierarchy:



\begin{corollary}
\label{cor:easy-var-hierarchy}
For every integer $k \geq 3$, $\spec_k \subsetneq \spec_{2k+3}$.
\end{corollary}
\begin{proof}
The strict inclusion follows from
$$
\spec_k  \subseteq 
\ntime{n^2 2^{kn}}  \subsetneq 
\ntime{2^{(k+1)n}}  \subseteq 
\spec_{2(k+1)+1}  =
\spec_{2k+3}.
$$
The first inclusion follows from Theorem~\ref{theo:spec-ntime}
and the third from Theorem~\ref{theo:ntime-spec}.
The second strict inclusion follows from
Cook's non-deterministic time hierarchy theorem~\cite[Theorem~3.2]{Cook73,AB09}.
\end{proof}

The following corollary shows that
to settle Asser's conjecture, it is sufficient
to consider sentences using three variables and binary relations.

Define the following class:
\begin{eqnarray*}
\cospec_3^{bin} & := &
\left\{
\begin{array}{l|l}
\bbN^+ - S &
\begin{array}{l}
S = \spec(\phi) \ \mbox{and} \ \phi \ \mbox{uses only}
\\ 
\mbox{three variables and binary relations}
\end{array} 
\end{array} 
\right\}
\end{eqnarray*}

\begin{corollary}
\label{cor:padding}
$\ne=\cone$ if and only if
$\cospec_3^{bin}\subseteq \spec$.
\end{corollary}
\begin{proof}
The ``only if'' direction is trivial.
The ``if'' direction is as follows.
Suppose $\cospec_3^{bin} \subseteq \spec$.
Since $\ntime {2^n} \subseteq \spec_3$ (and uses only binary relations),
this means that for every $A \in \ntime {2^n}$,
the complement $\bbN^+ -A \in \spec$, and hence, also $\bbN^+ - A \in \ne$.
By padding argument,
this implies that for every set $A \in \ne$,
the complement $\bbN^+ - A$ also belongs to $\ne$.
\end{proof}

To end this section, we present a slightly weaker result of Theorem~\ref{theo:spec-ntime},
i.e. $\spec_k \subseteq \ntime {n^2 2^{kn}}$, 
which is already sufficient to yield the hierarchy in Corollary~\ref{cor:easy-var-hierarchy}.
First, we show the following normalisation of first-order logic with $k$ variables.

\begin{proposition}
\label{prop:normalisation}
{\bf (Normalisation of first-order logic with $k$ variables)}
Each first-order sentence $\phi$ with at most $k$ distinct variables
$\vx = (x_1,\ldots,x_k)$ is equivalent to an existential second-order sentence
of the form: $\Phi := \exists R_1 \cdots \exists R_m \ \phi'$,
where each $R_i$ is a relation symbol of arity $\leq k$,
and $\phi'$ is a conjunction of first-order sentences with
variables $\vx = (x_1,\ldots,x_k)$ of either of the forms (1) and (2) below:
\begin{enumerate}[(1)]\itemsep=0pt
\item
$\forall x_1 \ \cdots \ \forall x_{k-1} \ \forall x_k \ \psi(x_1,\ldots,x_k)$,
\item
$\forall x_1 \ \cdots \ \forall x_{k-1} \ \exists x_k \ \psi(x_1,\ldots,x_k)$,
\end{enumerate}
where $\psi(x_1,\ldots,x_k)$ is a quantifier-free formula
in disjunctive normal form.
\end{proposition}
\begin{proof}
First, we assume that all the negations in $\phi$ are pushed
inside to the atomic formulae.

We associate each subformula $\theta(v_1,\ldots,v_q)$ of $\phi$,
where $0 \leq q \leq k$ and each $v_i \in \vx$,
including the sentence $\phi$, with a new relation symbol $R_{\theta}$
of arity $q$.
The relation symbol $R_{\theta}$ intuitively represents $\theta$.
Note that a relation symbol of arity $0$ is a Boolean variable
which can be either true or false.

The formula $\phi'$ is the conjunction of the atomic relation $R_{\phi}$
of arity $0$
and the formula $\delta_{\theta}$ corresponding to subformula $\theta(v_1,\ldots,v_q)$
of $\phi$ defined inductively as follows.
\begin{itemize}\itemsep=0pt
\item
If $\theta$ is a negation of an atomic formula $S(v_1,\ldots,v_q)$,
then 
\begin{eqnarray*}
\delta_{\theta} & := & 
\forall v_1 \cdots \forall v_q \ R_{\theta}(v_1,\ldots,v_q) \IFF
\neg S(v_1,\ldots,v_q).
\end{eqnarray*}
\item
If $\theta$ is of the form $\theta_1 \circledast \theta_2$, 
with free variables $v_1,\ldots,v_q$,
where $\circledast \in \{\wedge,\vee\}$
then 
\begin{eqnarray*}
\delta_{\theta} & := & 
\forall v_1 \cdots \forall v_q \ R_{\theta}(v_1,\ldots,v_q) \IFF
R_{\theta_1}(v_1,\ldots,v_q) \circledast R_{\theta_2}(v_1,\ldots,v_q).
\end{eqnarray*}
Note that if $\theta$ has no free variable, then
$\delta_{\theta}$ is $R_{\theta} \IFF R_{\theta_1} \circledast R_{\theta_2}$.
\item
If $\theta$ is $\forall v_q \ \theta'(v_1,\ldots,v_{q-1},v_q)$,
then 
\begin{eqnarray*}
\delta_{\theta} & := & 
\forall v_1 \cdots \forall v_{q-1} \  
R_{\theta}(v_1,\ldots,v_{q-1}) \IFF
\forall v_q \ 
R_{\theta'}(v_1,\ldots,v_q),
\end{eqnarray*}
which is equivalent to
\begin{eqnarray*}
\delta_{\theta} & := & 
\big(\forall v_1 \cdots \forall v_{q-1} \forall v_{q}
\ R_{\theta}(v_1,\ldots,v_{q-1}) \Rightarrow
R_{\theta'}(v_1,\ldots,v_q)
\big)
\ \wedge 
\\
& &
\big(\forall v_1 \cdots \forall v_{q-1} \exists v_{q} \
R_{\theta'}(v_1,\ldots,v_q)
\Rightarrow 
R_{\theta}(v_1,\ldots,v_{q-1})
\big).
\end{eqnarray*}
\item
If $\theta$ is $\exists v_q \ \theta'(v_1,\ldots,v_{q-1},v_q)$,
then 
\begin{eqnarray*}
\delta_{\theta} & := & 
\forall v_1 \cdots \forall v_{q-1} \  
R_{\theta}(v_1,\ldots,v_{q-1}) \IFF
\exists v_q \ 
R_{\theta'}(v_1,\ldots,v_q),
\end{eqnarray*}
which is equivalent to
\begin{eqnarray*}
\delta_{\theta} & := & 
\big(\forall v_1 \cdots \forall v_{q-1} \exists v_{q}
\ R_{\theta}(v_1,\ldots,v_{q-1}) \Rightarrow
R_{\theta'}(v_1,\ldots,v_q)
\big)
\ \wedge 
\\
& &
\big(\forall v_1 \cdots \forall v_{q-1} \forall v_{q} \
R_{\theta'}(v_1,\ldots,v_q)
\Rightarrow 
R_{\theta}(v_1,\ldots,v_{q-1})
\big).
\end{eqnarray*}
\end{itemize}
Note that in the definition above, if $\theta$ is an atomic formula,
then $R_{\theta}$ is $\theta$ itself.

Written formally,
\begin{eqnarray*}
\Phi & := & \exists R_1 \ \cdots \ \exists R_m
\ R_{\phi} \ \wedge \ \bigwedge_{\theta} \delta_{\theta},
\end{eqnarray*}
where $R_1,\ldots,R_m$ are all the $R_{\theta}$'s
and $\theta$ spans over all the subformulae of $\phi$.
It is straightforward to see that $\Phi$ and $\phi$ are equivalent.
\end{proof}

The following complexity result is an easy consequence of the normalisation lemma:

\begin{corollary}
\label{cor:spec-ntime}
For every positive integer $k$,
$\spec_k \subseteq \ntime {2^{kn}n^2}$.
\end{corollary}
\begin{proof}
By the above lemma, each first-order sentence $\phi$ using $k$ variables is
equivalent to the normalised formula $\Phi := \exists R_1 \cdots \exists R_m \ \phi'$.
By our construction, the quantification-depth of $\phi'$ is $k$.
Hence, on the domain $[N]$, where  $N=\Theta(2^n)$,
one can obtain a propositional Boolean formula $F_{\phi,N}$
with size $O(N^k)$\footnote{The size of a propositional Boolean formula is
the total sum of the number of appearances of each atom.},
such that $N \in \spec(\phi)$ if and only if $F_{\phi,N}$ is satisfiable.

It is well known that the satisfiability of problem
of a propositional Boolean formula $F$ of size $\ell$
with variables $p_i$ of indices $i \leq \ell$,
hence, of total length $|F| = O(\ell \log \ell)$
(in a fixed finite alphabet), can be solved in time $O(\ell \log^2 \ell)$
on a non-deterministic Turing machine.
We present it here in our specific case where,
as a straightforward consequence of Proposition~\ref{prop:normalisation},
the Boolean formula $F_{\phi,N}$ so obtained is 
a conjunction of DNF formulae, i.e.
of the form $F_{\phi,N} : C_1 \wedge \cdots C_m$,
and each $C_i$ is a DNF formula.
It is easy to see that the satisfiability problem
of formula in such a form can be decided by the following
non-deterministic algorithm:
\begin{itemize}\itemsep=0pt
\item
For each conjunct $C_i$, choose (non-deterministically)
a disjunct $\gamma_i$ of $C_i$.
Note that $\gamma_i$ is a conjunction of literals.
\item
Check deterministically whether
the conjunction $G:= \gamma_1 \wedge \cdots \wedge \gamma_m$
which is a conjunction $\ell_1 \wedge \cdots \wedge \ell_q$
of literals is satisfiable.
This can be done by sorting the list of literals $\ell_1,\ldots,\ell_q$
of $G$ in lexicographical order and checking that the sorted
list contains no pair of contiguous contradictory literals $p,\neg p$.
\end{itemize}
It is a folklore result that
a list of non-empty words $w_1,\ldots,w_q$ can be sorted
in lexicographical order on a multi-tape Turing machine in $O(\lambda \log \lambda)$,
where $\lambda = |w_1|+\cdots + |w_q|$.
Here, we have $\lambda = |G| \leq |F_{\phi,N}| = O(\ell \log \ell)$.
Altogether, it takes $O(\ell \log^2 \ell)$ time.
\end{proof}




\section{A finer hierarchy}
\label{sec:var-hierarchy}

In this section we are going to present a finer hierarchy of the spectra:
For every integer $k \geq 3$, $\spec_{k} \subsetneq \spec_{2k+2}$.
The outline of the proof follows the one in the previous subsection.

\begin{theorem}
\label{theo:ntime-spec-half}
For every integer $k \geq 2$, $\ntime {2^{(k+\frac{1}{2})n}} \subseteq \spec_{2k+2}$.
\end{theorem}
\begin{proof}
We follow the outline of Proposition \ref{prop:three-var} and Theorem
\ref{theo:ntime-spec}. Now $M$ is an NTM 
that accepts $\bin(N)$ in time
$N^kR$ and space $N^kR$, where $R = \lfloor \sqrt{N-1}\rfloor$.
The space-time diagram of the computation of $M$ is then
depicted as an $[N^k\cdot R]\times [N^k\cdot R]$ grid.

Each point in $[N^k\cdot R]\times [N^k\cdot R]$ grid
can be identified as a point in $[N]^k\times [R]\times [N]^k\times [R]$.
By the converse of the pairing function $(r)\mapsto (\pi_x(r),\pi_y(r))$,
where $\pi_x(r) = r \bmod R = r_1$, and $\pi_y(r) = \lfloor(r / R)\rfloor = r_2$,
each point in $((\vx,r_1),(\vy,r_2)) \in [N]^k\times [R]\times [N]^k\times [R]$
can be represented as $(\vx,\vy,r) \in [N]^k\times [N]^k \times [N]$,
where $r= r_1 +r_2R$.

So the computation of $M$ can be viewed as labelling of the point
$(\vx,\vy,r) \in [N]^k\times [N]^k \times [N]$.
The only difference now is we need to define the shift operators 
$\Delta_h^r$, $\overline{\Delta}_h^r$ and $\Delta_v^r$ -- 
the analog of the shift operators
$\Delta_h$, $\overline{\Delta}_h$ and $\Delta_v$, respectively, 
in the proof of Theorem~\ref{theo:ntime-spec}.


As previously, we define the order $<$, minimum $\relMIN$, maximum $\relMAX$,
and the induced successor relation $\suc$.
We also define the following relations:
\begin{itemize}\itemsep=0pt
\item 
$\relADD(x,y,z)$ which holds if and only if $x+y=z$.
$$
\forall x \forall y \forall z
\left(
\begin{array}{l}
\relADD(x,y,z) \ \IFF
\\
\left(
\begin{array}{l}
(\relMIN(y) \wedge x=z) \ \vee 
\\
(\exists y'  \ \exists z' \ \suc(y',y)
 \wedge \suc(z',z) \wedge \relADD(x,y',z'))
\end{array}
\right)
\end{array}
\right)
$$
\item 
$\relMUL(x,y,z)$ which holds if and only if $xy=z$.
$$
\forall x 
\forall y 
\forall z 
\left(
\begin{array}{l}
\relMUL(x,y,z) \ \IFF
\\
\left(
\begin{array}{l}
(\relMIN(y) \wedge \relMIN(z)) \ \vee
\\
\exists y' \ \exists z' \
(\suc(y',y) \wedge \relMUL(x,y',z') \wedge \relADD(z',x,z) \end{array}
\right)
\end{array}
\right)
$$

\item 
$\relISR(x)$ which holds if only if $x = R$.
$$
\forall x \
\left(
\begin{array}{ccc}
\relISR(x) 
& \IFF &
\Big(
\exists y\ \relMUL(x,x,y) \wedge \neg \exists x' \exists y' \ x'>x \wedge \relMUL(x',x',y')
\Big)
\end{array}
\right)
$$

\item 
$\relLESSR(x)$ which holds if only if $x < R$.
$$
\forall x \
\left(
\begin{array}{ccc}
\relLESSR(x) 
& \IFF &
\exists y \ y>x \wedge \relISR(y)
\end{array}
\right)
$$

\item 
$\relLESSRR(x)$ which holds if only if $x < R^2$.
$$
\forall x \
\left(
\begin{array}{ccc}
\relLESSRR(x) 
& \IFF &
\exists y \ \exists z \ 
\relISR(y) \ \wedge \ \relMUL(y,y,z) \ \wedge \ x < z
\end{array}
\right)
$$

\item 
$\relPROJECT(r,x,y)$ which holds if only if 
$x = \pi_x(r) = r \ \mbox{mod} \ R$ and $y=\pi_y(r) = \floor {r/R}$.
$$
\forall r
\forall x
\forall y 
\left(
\begin{array}{l}
\relPROJECT(r,x,y) \ \IFF
\\
\left(
\begin{array}{l}
\relLESSRR(r) \wedge \relLESSR(x) \wedge \relLESSR(y) \wedge
\\
\exists z \ \exists z' \
(\relISR(z') \wedge \relMUL(y,z',z) \wedge \relADD(x,z,r)
\end{array}
\right)
\end{array}
\right)
$$
\end{itemize}
Using the relations above, it is straightforward
to write the definitions below as first-order axioms using at most five variables:
\begin{itemize}\itemsep=0pt
\item 
Cyclic successor in $[R]$: 

$\relcyc(x,y)$ if and only if $x,y \in [R]$, and either $y=x+1$, or $x=R-1$ and $y=0$.
\item 
Horizontal successor in $[R^2]$: 

$\relsucx(r,r')$ if and only if 
$r,r' \in [R^2]$, $\pi_y(r) = \pi_y(r')$ and $\relcyc(\pi_x(r), \pi_x(r'))$.
\item 
Vertical successor in $[R^2]$: 

$\relsucy(r,r')$ if and only if $r,r' \in [R^2]$, $\pi_x(r) = \pi_x(r')$ and $\relcyc(\pi_y(r), \pi_y(r'))$.
\item 
Horizontal minimum in $[R^2]$:

$\relMINx(r)$ if and only if $r \in [R^2]$ and $\pi_x(r) = 0$.

\item 
Vertical minimum in $[R^2]$: 

$\relMINy(r)$ if and only if $r \in [R^2]$ and $\pi_y(r) = 0$.
\end{itemize}
All the definitions above use at most five variables, which is $\leq 2k+2$, for each integer $k\geq 2$.

The operators $\Delta^r_h \phi$, $\overline{\Delta}^r_h \phi$ and $\Delta^r_v \phi$ are defined as follows.
\begin{eqnarray*}
\Delta^r_h \phi(\vx,\vy,r) &:=& 
\forall z 
\left(
\begin{array}{l}
\relsucx(r,z) \Rightarrow 
\left(
\begin{array}{lll}
(\relMINx(z) &\wedge& \Delta_h \phi(\vx,\vy,z)) \ \vee \\
(\neg \relMINx(z) &\wedge& \phi(\vx,\vy,z)))
\end{array}
\right) 
\end{array}
\right)
\\
\overline{\Delta}^r_h \phi(\vx,\vy,r) &:=& 
\forall z 
\left(
\begin{array}{l}
\relsucx(z,r) \Rightarrow 
\left(
\begin{array}{lll}
(\relMINx(r) &\wedge& \overline{\Delta}_h \phi(\vx,\vy,z)) \ \vee \\
(\neg \relMINx(r) &\wedge& \phi(\vx,\vy,z)))
\end{array}
\right) 
\end{array}
\right)
\\
\Delta^r_v \phi(\vx,\vy,r) & := &
\forall z 
\left(
\begin{array}{l}
\relsucy(r,z) \Rightarrow 
\left(
\begin{array}{lll}
(\relMINy(z) &\wedge& \Delta_v \phi(\vx,\vy,z)) \ \vee \\
(\neg \relMINy(z) &\wedge& \phi(\vx,\vy,z)))
\end{array}
\right) 
\end{array}
\right)
\end{eqnarray*}
where $\Delta_h$ is to access the successor of $\vx$,
$\overline{\Delta}_h$ the predecessor of $\vx$
and $\Delta_v$ the successor of $\vy$.
They are all defined just like in the proof of Theorem~\ref{theo:ntime-spec}.
This completes our proof of Theorem~\ref{theo:ntime-spec-half}.
\end{proof}

Now, combining both 
Theorems~\ref{theo:ntime-spec-half} and~\ref{theo:spec-ntime},
as well as the argument in the proof of Corollary~\ref{cor:easy-var-hierarchy},
we obtain that:
$$
\spec_k  \subseteq 
\ntime{n^2 2^{kn}}  \subsetneq 
\ntime{2^{(k+\half)n}}  \subseteq 
\spec_{2k+2},
$$
hence, establishing the following hierarchy.

\begin{corollary}
\label{cor:var-hierarchy}
For every integer $k \geq 3$, $\spec_k \subsetneq \spec_{2k+2}$.
\end{corollary}

\section{Translating our results to classes $\np$ and $\soe$}
\label{sec:np-soe}

In this section we are going to show how our results 
can be translated into relations between the class $\np$
and the class of existential second-order sentences $\soe$.
We provide a brief review of their definitions here.
For more details, we refer the reader to Immerman's textbook~\cite{Immerman99}.

Let $\soe$ denote the class of existential second-order sentences.
A sentence $\Phi \in \soe$ defines a class of structures
$\{\stra \mid \stra \models \Phi\}$.
A celebrated result of Fagin states that
$\soe = \np$, where
the input to the $\np$ Turing machine is the binary encoding
of the structures.

Let $\soevar {k}$ be the class $\soe$ where the first-order sentences 
uses only up to $k$ variables.
Now, Theorems~\ref{theo:ntime-spec} and~\ref{theo:ntime-spec-half} 
can be respectively rewritten as:
\begin{eqnarray}
\mbox{For any integer} \ k\geq 1,
& & \ntime{n^k} \subseteq \soevar {2k + 1}
\label{eq1}
\\
\mbox{For any integer} \ k\geq 2,
& & \ntime{n^{k+1/2}}\subseteq \soevar {2k + 2}
\label{eq2}
\end{eqnarray}
Indeed, let $M$ be a non-deterministic Turing machine accepting a
binary language $L$ within time $O(n^k)$, where $n$ is the length of the input string 
$w =w_0\ldots w_{n-1} \in \{0,1\}^*$ . The input can be viewed as a structure over $[n]$ 
with vocabulary the binary successor relation $\suc$ and the unary predicate $S$, 
where $S(x)$ holds if and only if $w_x =1$.

The formula $\Phi$ constructed in the proof of Theorem~\ref{theo:ntime-spec} 
(resp. Theorem~\ref{theo:ntime-spec-half}) can be viewed as an 
$\soevar {2k + 1}$ (resp. $\soevar {2k + 2}$) formula,
where the predicates $\relSYMBOL^i_a$ and $\relSTATE^i_q$, as well as $\relDOUBLE$,
$\relHALF$, $\relDIV$, $\relBIT$ , etc. are existentially quantified.

On the other hand, Theorem~\ref{theo:spec-ntime} can be rewritten as:
\begin{eqnarray}
\label{eq3}
\soevar {k} \subseteq \ntime {n^k \log^2 n}
\end{eqnarray}
Equations~\ref{eq3} and~\ref{eq2} then yield the chain of inclusions:
$$
\soevar {k} \ \subseteq \ \ntime{n^k \log^2 n}
\ \subsetneq \
\ntime {n^{k+1/2}} \ \subseteq \ \ \soevar{2k + 2}
$$
and hence, $\soevar {k} \subsetneq \soevar {2k + 2}$, 
for each $k\geq 3$.

\section{Concluding remarks}
\label{sec:conclusion}

In this paper we present two results
that we believe contribute to our understanding of the spectra problem.
The first is that there is an
infinite hierarchy of first-order spectra based on the number of variables:
$\spec_{k} \subsetneq \spec_{2k+2}$.
The proof is based on tight relationships between the class $\ne$
and first-order spectra $\spec$.

The second result is  
that to settle Asser's conjecture it is sufficient
to consider sentences using three variables and binary relations.
This seems to be the furthest we can go.
As mentioned in Section~\ref{sec:related-works},
we recently showed that the class of spectra of 
two-variable logic with counting quantifiers are exactly semilinear sets,
and closed under complement~\cite{KT}.

\subsection*{Acknowledgement}
We would like to thank Etienne Grandjean for his careful reading
and in providing most of the literature pointers for Section~\ref{sec:related-works},
and in suggesting Proposition~\ref{prop:normalisation} and
Corollary~\ref{cor:spec-ntime}.
His very extensive comments have guided and helped us in improving our manuscript.
We also thank the anonymous reviewer for his careful review and
Ron Fagin for some discussions on related results.
The first author is supported by the Polish National Science Centre Grant
DEC - 2012/07/D/ST6/02435.
The second author is supported by FWO Pegasus Marie Curie fellowship.

\bibliographystyle{acmtrans}
\bibliography{spec-k}

\end{document}